\documentclass{article}
\usepackage{arxiv}

\usepackage{hyperref}
\usepackage{microtype}

\usepackage{amssymb}
\usepackage{amsmath}
\usepackage{amsthm}
\usepackage{mathdots}
\usepackage{mathtools}

\usepackage{tensor}

\usepackage{urwchancal}
\DeclareFontFamily{OT1}{pzc}{}
\DeclareFontShape{OT1}{pzc}{m}{it}{<-> s * [1.10] pzcmi7t}{}
\DeclareMathAlphabet{\mathpzc}{OT1}{pzc}{m}{it}

\usepackage{graphicx}
\usepackage{ifpdf}
\ifpdf
\usepackage{epstopdf}
\PrependGraphicsExtensions{.svg}
\fi

\usepackage{xypic}

\newtheorem{theorem}{Theorem}[section]
\newtheorem{lemma}[theorem]{Lemma}
\newtheorem{corollary}[theorem]{Corollary}

\newtheorem{definition}[theorem]{Definition}

\newtheorem{remark}[theorem]{Remark}

\newenvironment{List}
{\begin{list}
 {[\arabic{enumi}]}
 {\usecounter{enumi}
  \settowidth{\labelwidth}{[99]}}
  \setlength{\labelsep}{1em}}
{\end{list}}

\providecommand{\N}{\mathbb{N}}
\providecommand{\R}{\mathbb{R}}

\providecommand{\grpG}{\mathbf{G}}

\providecommand{\gothg}{\mathfrak{g}}

\providecommand{\gothz}{\mathfrak{z}}
\providecommand{\gothB}{\mathfrak{B}}

\providecommand{\gothL}{\mathfrak{L}}
\providecommand{\gothR}{\mathfrak{R}}

\providecommand{\gothX}{\mathfrak{X}} 

\providecommand{\calF}{\mathcal{F}}

\providecommand{\calM}{\mathcal{M}}
\providecommand{\calN}{\mathcal{N}}

\providecommand{\vecV}{\mathbb{V}}

\providecommand{\PD}{\mathbb{S}_+} 

\providecommand{\Id}{I} 

\DeclareMathOperator{\Ad}{Ad}

\DeclareMathOperator{\kernel}{ker}
\DeclareMathOperator{\image}{im}

\providecommand{\td}{\mathrm{d}}
\providecommand{\tD}{\mathrm{D}}

\providecommand{\ddt}{\frac{\td}{\td t}}

\usepackage{accents}
\makeatletter
\providecommand{\scirc}{%
    \hbox{\fontfamily{\rmdefault}\fontsize{0.4\dimexpr(\f@size pt)}{0}\selectfont{\raisebox{-0.52ex}[0ex][-0.52ex]{$\circ$}}}}

\makeatother

\mathchardef\mhyphen="2D

\newcommand{\publicationdetails}
{to appear MTNS2020 (delayed until August 2021)}
\newcommand{\publicationversion}
{Preprint}

\begin{document}

\title{Equivariant Filter Design for Kinematic Systems on Lie Groups}
\headertitle{Equivariant Filter Design for Kinematic Systems on Lie Groups}

\author{
\href{https://orcid.org/0000-0002-7803-2868}{\includegraphics[scale=0.06]{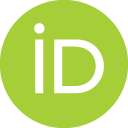}\hspace{1mm}
Robert Mahony}
\\
	Systems Theory and Robotics Group \\
	Australian National University \\
    ACT, 2601, Australia \\
	\texttt{Robert.Mahony@anu.edu.au} \\
	\And	\href{https://orcid.org/0000-0002-5881-1063}{\includegraphics[scale=0.06]{orcid.png}\hspace{1mm}
Jochen Trumpf} \\
	Software Innovation Institute \\
	Australian National University\\
    ACT, 2601, Australia \\
	\texttt{Jochen.Trumpf@anu.edu.au} \\
}
\maketitle

\begin{abstract}
It is known that invariance and equivariance properties for systems on Lie groups can be exploited in the design of high performance and robust observers and filters for real-world robotic systems.
This paper proposes an analysis framework that allows any kinematic system on a Lie group to be embedded in a natural manner into an equivariant kinematic system.
This framework allows us to characterise the properties of, and relationships between, invariant systems, group affine systems, and equivariant systems.
We propose a new filter design, the \emph{Equivariant Filter} (EqF), that exploits the equivariance properties of the system embedding and can be applied to any kinematic system on a Lie group.
\end{abstract}


\section{Introduction}

Systems on Lie groups have been studied since the early 1970s with early work mostly due to \cite{brockett1972,brockett1973} and \cite{jurdjevic1972}.
The initial motivation was the dynamic modelling of mechanical systems for control purposes, see \cite{brockett1976}.
More recently \cite{aghannan2003} considered observer design for mechanical systems by exploiting symmetry properties.

The control of aerial robotic systems, and particularly of quadrotor vehicles, requires good estimates of the vehicle's attitude and led to a focus on observer design for kinematic systems on Lie groups where the velocity is measured, cf. \cite{2008_Bonnabel_TAC,RM_2008_MahHamPfl.TAC,lageman2009}.
The success of these algorithms in providing high performance observers for real world systems has led to more general studies of invariance principles and error dynamics, cf.  \cite{bonnabel2009,lageman2010,trumpf2012,RM_2013_Mahony_nolcos}.
More recent work on the structure of kinematic systems on Lie groups was published by \cite{2018_Trumpf_CDC}.

The nature of the possible observer designs is closely linked to the invariance properties of the system considered.
Gradient and explicit designs, cf. \cite{2008_Bonnabel_TAC,RM_2008_MahHamPfl.TAC,lageman2009}, tend to depend on what is termed Type I invariance, cf. \cite{RM_2013_Mahony_nolcos}, for which the error dynamics of the observer are autonomous, cf. \cite{lageman2010}.
A more sophisticated filter design, the Invariant Extended Kalman Filter (IEKF) was introduced by \cite{2009_Bonnabel_cdc} for systems with both Type I and Type II invariant velocities.
In recent work \cite{2017_Barrau_tac} provided an explicit algebraic invariance condition termed \emph{group affine} that characterises the systems to which IEKF design can be applied.

To the best of the authors' knowledge, the precise relationship between group affine systems and (classically) invariant systems has not previously been explored in the literature.

In this paper we study invariance and equivariance properties for systems on Lie groups in the context of observer design.

\begin{figure}[h]
\begin{center}
\includegraphics[scale=0.4]{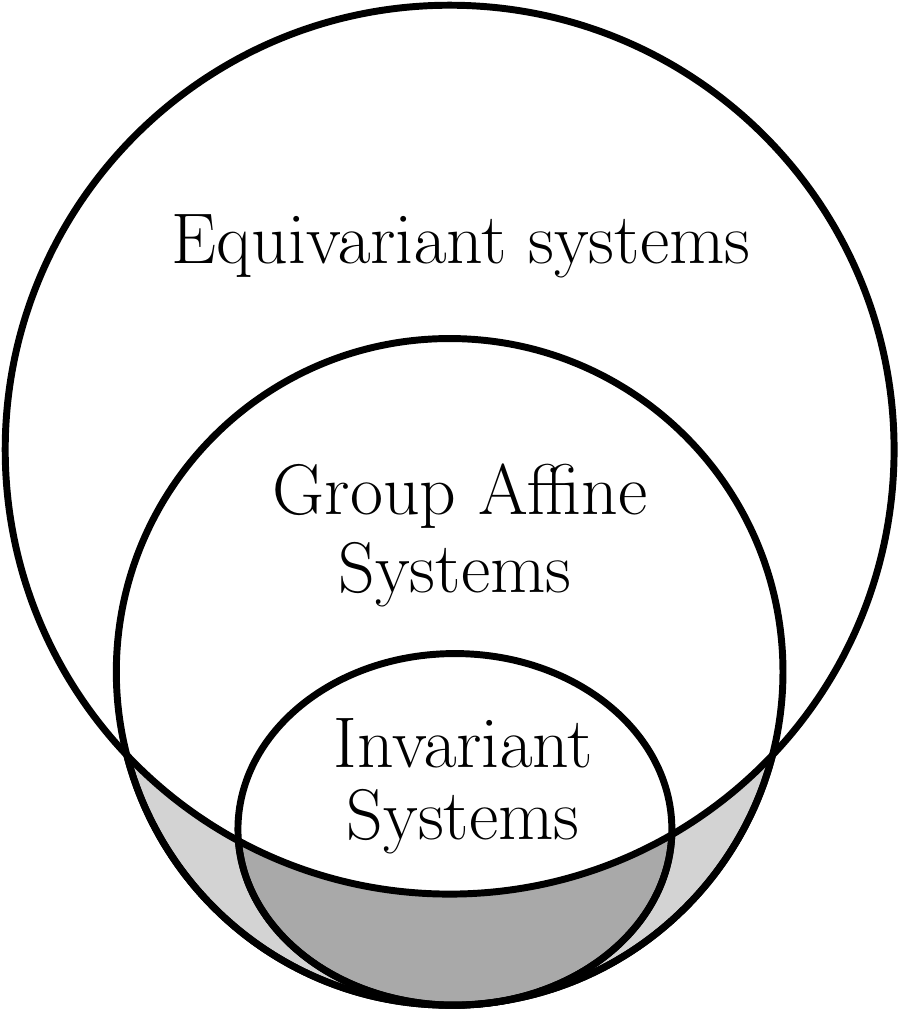}
\caption{Hierarchy of kinematic systems on Lie groups.}
\label{fig:bifurcation}
\end{center}
\end{figure}

There are three key results in this paper:
\begin{itemize}
\item We show that any kinematic system on a Lie group can be embedded into an equivariant kinematic system by a suitable extension of the input vector space.

\item We characterise invariant systems on Lie groups using a direct sum decomposition of the input vector space. Invariant systems are shown to be a subset of group affine systems. The equivariant embedding of group affine systems is shown to only require a finite dimensional input vector space extension.

\item We propose the \emph{Equivariant Filter} (EqF), a new observer design methodology that can be applied to any equivariant system which, considering the first contribution above, means it can be applied to any kinematic system on a Lie group.
\end{itemize}

The approach taken is to model a kinematic system as a linear subspace of the vector space of smooth vector fields on the Lie group parameterized by the input vector space, see Section~\ref{sec:kinematics}.
We show that there is a natural group action on vector fields induced by right translation (Lemma~\ref{lem:d_star_R} in Section~\ref{sec:prelim}) and then go on to define an extended input space as the smallest vector subspace of the space of smooth vector fields that is closed under this action and contains the original system, see Section~\ref{sec:equiv_systems}.
It turns out that this construction is sufficient to ensure equivariance of the resulting extended system (Theorem~\ref{th:equiv_ext}).
We exploit the equivariance property to study first invariant systems, cf. \cite{RM_2013_Mahony_nolcos}, in Section~\ref{sec:invariant_systems} and then group affine systems, cf. \cite{2017_Barrau_tac}, in Section~\ref{sec:GA_systems}.
We show that the equivariant input extension of an invariant system is invariant (Corollary~\ref{cor:right_invariant} and Lemma~\ref{lem:left_invariant_input_ext}) and that the input vector space of an invariant system can always be decomposed into a direct sum of three components, a Type I, a Type II, and a Type 0 component, where the Type 0 component has both Type I and Type II invariance properties (Theorem~\ref{th:inv_sys_decomp}).
We furthermore show that every invariant system is group affine (Lemma~\ref{lem:inv_implies_GA}) and that the equivariant input extension of a group affine system is group affine and only differs from the original system by a finite dimensional Type I invariant velocity subspace (Theorem~\ref{th:GA_ext}).

Finally, we consider the observer design problem in Section~\ref{sec:EqF} and show that the equivariant structure can be exploited to derive a Kalman-Bucy like filter for any equivariant system on a Lie group.
The proposed Equivariant Filter (EqF), see Equations~\eqref{eq:EqF} and \eqref{eq:EqF_Sigma}, specialises to the Invariant Extended Kalman Filter (IEKF) in the case that the system is group affine or invariant.

\section{Preliminaries}\label{sec:prelim}

The explicit derivative of a function $f : \calM \to \calN$ evaluated at $x_0$ is written
\[
\left. \tD_x \right|_{x_0} f(x) : T_{x_0} \calM \to T_{f(x_0)} \calN.
\]
The same notation is used for the partial derivative $\left. \tD_x \right|_{x_0} f(x,y)$ where the $y$ variable is held constant in the differentiation.
The differential of a function is denoted $\td f : T \calM \to T \calN$ with $\td f \cdot \eta_{x_0} = \left. \tD_x \right|_{x_0} f(x) [\eta_{x_0}]$ where $\eta_{x_0} \in T_{x_0} \calM$.

A Lie group $\grpG$ is both a group and a smooth manifold for which the group multiplication and inverse are smooth in the differential structure on the manifold.
In this paper we will restrict our attention to matrix Lie groups, although the results will naturally generalise to any finite dimensional semi-algebraic Lie group and many of the results will be more general again.
We write group multiplication as $AB$ and group inverse as $A^{-1}$ for $A, B \in \grpG$ and think of these as matrix multiplication and inverse.
For an element $A \in \grpG$ define right (resp.~left) multiplication maps $R_A: \grpG \to \grpG$, $R_A(X) := XA$ (resp.~$L_A(X) := AX$).

Let $\gothX(\grpG)$ denote the set of all smooth vector fields on $\grpG$ and note that $\gothX(\grpG)$ is an infinite dimensional vector space.
We will use $\gothg$ to denote the tangent space\footnote{The Lie-algebraic structure of $\gothg$ is not exploited in the present paper.}
 $T_{\Id} \grpG$ of $G$ at the identity.
In a matrix Lie group then $\td R_X \Lambda=\Lambda X$ (resp.~$\td L_X \Lambda=X\Lambda$) for all $X\in\grpG$ and $\Lambda\in\gothg$.
Define
\[
\gothL(\grpG) = \{ \td L_\grpG \Lambda\colon X\mapsto X \Lambda \;|\; \Lambda \in \gothg\} \subset \gothX(\grpG)
\]
to be the set of \emph{left invariant} vector fields.
The mapping $\Lambda \mapsto \td L_\grpG \Lambda$ establishes an isomorphism from $\gothg$ to $\gothL(\grpG)$ and the two spaces are often identified in the literature.
In an analogous fashion, the set of right invariant vector fields $\{ \td R_\grpG \Lambda \;|\; \Lambda \in \gothg \}$ also forms a vector subspace $\gothR(\grpG) \subset \gothX(\grpG)$ isomorphic to $\gothg$.
In this paper, we will distinguish strongly between $\gothg = T_\Id \grpG$ and the two vector spaces $\gothL(\grpG) \subset \gothX(\grpG)$ and $\gothR(\grpG) \subset \gothX(\grpG)$.
The intersection $\gothB(\grpG) := \gothL(\grpG) \cap \gothR(\grpG) \subset \gothX(\grpG)$ is the finite dimensional vector subspace of bi-invariant vector fields.
For a connected Lie group $\grpG$ the bi-invariant vector fields are parametrized by the center $\gothz \lhd \gothg$ via
\[
\gothB(\grpG) = \{ \td L_\grpG \Lambda = \td R_\grpG \Lambda \;|\; \Lambda \in \gothz \}.
\]

A right group action $\psi$ of $\grpG$ on a vector space $V$ is a mapping
\begin{align*}
\psi & \colon \grpG \times V \rightarrow V,
\end{align*}
with $\psi(A,\phi(B,v)) = \psi(BA,v)$ and $\psi(I,v) = v$
for all $A,B\in\grpG$ and $v\in V$.
The action $\psi$ is called \emph{linear} if it induces linear mappings $\psi_A \colon V \rightarrow V$ for $A \in \grpG$ by $\psi_A(v) := \psi(A,v)$.

\begin{lemma}\label{lem:d_star_R}
Define a smooth map $\td_\star R : \grpG \times \gothX(\grpG) \rightarrow \gothX(\grpG)$, by
\begin{align}
{\td_\star R} (Z, F) := \td R_Z \cdot F \circ R_{Z^{-1}} \in \gothX(\grpG).
\label{eq:overline_dphi}
\end{align}
Then ${\td_\star R}$ is a linear group action on the vector space $\gothX(\grpG)$.
\hfill$\square$
\end{lemma}

\begin{proof}
Note that for $A\in\grpG$, $R_A$ is a diffeomorphism with smooth inverse $R_{A^{-1}}$.
We have
\begin{align*}
{\td_\star R} (A, {\td_\star R} (B, F)) &=
\td R_A \cdot (\td R_B \cdot F \circ R_{B^{-1}})\circ R_{A^{-1}} \\
& = \td R_{BA} \cdot F \circ R_{{BA}^{-1}} \\
& = {\td_\star R} (BA, F)
\end{align*}
for all $A,B\in\grpG$ and $F\in\gothX(\grpG)$ since $R$ is a right action on $\grpG$.
The identity property of the group action is straightforward and this demonstrates $\td_\star R$ is a group action.
Linearity follows since
\begin{align*}
{\td_\star R} (A, \alpha_1 F_1 & + \alpha_2 F_2) =
\td R_A \cdot (\alpha_1 F_1 + \alpha_2 F_2 ) \circ R_{A^{-1}} \\
& = \alpha_1 \td R_A \cdot  F_1\circ R_{A^{-1}} +
\alpha_2\td R_A \cdot F_2 \circ R_{A^{-1}} \\
& = \alpha_1 {\td_\star R} (A, F_1) + \alpha_2 {\td_\star R} (A, F_2)
\end{align*}
for all $A\in\grpG$, $F_1,F_2\in\gothX(\grpG)$ and $\alpha_1,\alpha_2\in\R$.
\end{proof}

\begin{remark}
The linear group action ${\td_\star R}$ defines a representation of the Lie group $\grpG$ on the infinite dimensional vector space $\gothX(\grpG)$. We will see that equivariant kinematic systems correspond precisely to subrepresentations of this representation.
\hfill$\square$
\end{remark}

\section{Problem formulation}\label{sec:kinematics}

Let $\grpG$ be a Lie group and $V$ be a vector space.
The system function for a kinematic system on $\grpG$ is a \emph{linear} map
\begin{align}
\calF \colon & V \rightarrow \gothX(\grpG)  \\
& v \mapsto \calF_v \notag
\end{align}
where $\gothX(\grpG)$ is the vector space of smooth vector fields on $\grpG$.
Trajectories of the system are given by the ordinary differential equation
\begin{align}\label{eq:system}
\dot{X} = \calF_{v(t)}(X), \quad\quad X(0) \in \grpG
\end{align}
for initial conditions $X(0) \in \grpG$ and input signals $v(t) \in V$.

Without loss of generality we will assume that $\kernel \calF = \{0\} $.
That is, that there is no $v \in V$, $v\ne 0$ such that $\calF_v \equiv 0$ on $\grpG$.
Since $\calF$ is linear with trivial kernel, then $\calF$ is an injection of $V \hookrightarrow \gothX (\grpG)$ onto its image $\image \calF \subset \gothX(\grpG)$.
That is, we can think of $\image \calF$ as a vector subspace of $\gothX(\grpG)$ parametrized by the inputs $v \in V$.

This paper is concerned with understanding various relationships between different properties of kinematic systems on $\grpG$.
A key observation is that some of these system properties are captured in the parametrization of the subspace $\image \calF$ and how it embeds into $\gothX(\grpG)$.

If $V^j \subset V$ is a vector subspace of $V$ then $\calF^j \colon V^j \to \gothX(\grpG)$ is a well defined system function and $\image \calF^j \subset \image \calF \subset \gothX(\grpG)$.
The trajectories of the kinematic system defined by $\calF^j$ form a sub-behaviour of those of the kinematic system defined by $\calF$ and we will use this approach to decompose more complex kinematics of $\calF$ into sub-behaviours with simple kinematics.

\begin{remark}
If the subspace $V^j$ has a complement in $V$, for example in the case of finite dimensional $V$, we can intuitively think of the sub-behaviour given by $V^j$ as consisting of those trajectories of the system where all inputs other than those in $V^j$ are set to zero.
\hfill$\square$\end{remark}

We introduce an algebraic object that is closely related to the system function $\calF$ and is important in understanding invariant system structures.

\begin{definition}\label{def:vel_lift}
Let $\calF \colon V \to \gothX(\grpG)$ be a kinematic system on a Lie group $\grpG$ over a vector space $V$.
The \emph{lift} is the function $\Lambda : \grpG \times V \to \gothg$ defined by
\begin{align}
\Lambda(X,v) := X^{-1} \calF_v(X) \label{eq:lift_function}
\end{align}
for all $X\in\grpG$ and $v\in V$.
\hfill$\square$
\end{definition}

The lift provides an algebraic structure that connects the input vector space $V$ to the Lie algebra $\gothg$.
Properties of equivariance, invariance and being group affine can be expressed as algebraic properties of the map $\Lambda$ that hold on vector subspaces of $V$.

We will also look at embedding the trajectories of a given kinematic system as a sub-behaviour of a higher dimensional kinematic system with desirable properties.

\section{Equivariant Systems}\label{sec:equiv_systems}

A (right) equivariant system is one in which right translation of the system function is the same as evaluating the function at the translated base point along with a possible group action transformation of the input space.
In this section we show that any kinematic system on a Lie group can be embedded in an equivariant system by extending the input space.

\begin{definition}\label{def:equi_sys}\textbf{(Equivariant System)}
A kinematic system $\calF \colon V \rightarrow \gothX(\grpG)$ is \emph{(right) equivariant} if there exists a right group action $\psi : \grpG \times V \rightarrow V$ such that
\begin{align}\label{eq:equi_F}
\td R_A \calF_v (X) = \calF_{\psi_A(v)}(R_A(X))
\end{align}
for all $A,X\in\grpG$ and $v\in V$.
\hfill$\square$
\end{definition}

Assume that a system $\calF\colon V \to \gothX(\grpG)$ is equivariant according to Definition \ref{def:equi_sys}.  Then
\begin{align}
\calF_{\psi_A(v)} (X)
& =  \calF_{\psi_A(v)} (R_A (X A^{-1})) \notag \\
& = \td R_A \calF_v (X {A^{-1}} ) \label{eq:Fpsi:2} \\
& = \td R_A\cdot \calF_v \circ R_{A^{-1}}(X) \notag \\
& = \td_\star R_A \calF_v (X) \label{eq:Fpsi:3}
\end{align}
where \eqref{eq:Fpsi:2} follows from \eqref{eq:equi_F} and \eqref{eq:Fpsi:3} follows from \eqref{eq:overline_dphi}.
In particular, the input group action $\psi$ is uniquely determined by the group action $\td_\star R$ on vector fields (Lemma~\ref{lem:d_star_R}).

Another way of reading \eqref{eq:Fpsi:3} is that for an equivariant kinematic system, the vector subspace $\image \calF\subset\gothX(\grpG)$ is invariant under the group action $\td_\star R$.
Conversely, given such an invariant subspace $V\subset\gothX(\grpG)$, define a kinematic system by the natural injection $\calF\colon V\hookrightarrow\gothX(\grpG)$. The invariance of $V$ then implies that for all $A\in\grpG$ and $v\in V$ we have $\td_\star R_A \calF_v=\calF_w$ for some $w\in V$. Define $\psi_A(v):=w$ and observe that this defines a right group action of $\grpG$ on $V=\image\calF\subset\gothX(\grpG)$ since $\td_\star R$ is a right group action on $\gothX(\grpG)$. The following lemma sums up this observation.

\begin{lemma}\label{lem:equivariance}
A kinematic system $\calF \colon V \rightarrow \gothX(\grpG)$ is equivariant if and only if the vector subspace $\image \calF\subset\gothX(\grpG)$ is invariant under the group action $\td_\star R$.
\hfill$\square$
\end{lemma}

\begin{remark}
Since subrepresentations of the representation $\td_\star R$ of $\grpG$ on $\gothX(\grpG)$ are by definition the restrictions of $\td_\star R$ to its invariant subspaces, the above analysis shows that equivariant kinematic systems are in direct correspondence with subrepresentations of the representation $\td_\star R$ of $\grpG$ on $\gothX(\grpG)$.
\hfill$\square$
\end{remark}

A key contribution of this paper is to show that it is possible to embed any kinematic system $F \colon \vecV \to \gothX(\grpG)$ in an equivariant system structure by extending the input space to include the image of $\td_\star R$ acting on $\image F$.

\begin{definition}\label{def:input_ext}\textbf{(Equivariant Input Extension)}
Let $F \colon \vecV \to \gothX(\grpG)$ be a kinematic system on a Lie group over a vector space $\vecV$.
Define
\begin{align}
V = \text{span} \{ \td_\star R_A F_v \;|\; v \in \vecV, A \in \grpG \} \subset \gothX(\grpG) \label{eq:input_extension}
\end{align}
to be the smallest vector subspace of $\gothX(\grpG)$ generated by the image of $\td_\star R$ applied to $\image F$ and define the associated extended kinematic system by the natural injection $\calF\colon V\hookrightarrow\gothX(\grpG)$.
\hfill$\square$
\end{definition}

The vector space $V \subset \gothX(\grpG)$ may be infinite dimensional, even if $\vecV$ is finite dimensional, depending on the nature of the kinematics function $F$ and how it interacts with the group action $\td_\star R$. Since we always assume that the system function $F$ is injective, we can think of $\vecV\equiv\image F\subset V$ as a vector subspace of $V$.

We will continue to use lower case letters $u, v, w \in V$ to refer to elements in $V$ to emphasise its input vector space structure and role.
Since every element of $V$ lies in the span of the vector fields $\td_\star R_A F_v\in\gothX(\grpG)$ then for every $w \in V$
\begin{align}
\calF_w = \sum_{i = 1}^{n_w} \td_\star R_{A_i} F_{v_i}
\label{eq:calF}
\end{align}
for some $n_w \in \N$, $A_i \in \grpG$ and $v_i \in \vecV$.
When the particular element considered $v \in \vecV \subset V$ then $\calF_v = F_v$ is given by the original definition of $F$.

The following theorem justifies the name \emph{equivariant} input extension.

\begin{theorem}\label{th:equiv_ext}
Let $F : \vecV \to \gothX(\grpG)$ be a kinematic system on a Lie group $\grpG$ over a vector space $\vecV$ and let $\calF\colon V\hookrightarrow\gothX(\grpG)$ be its equivariant input extension (Def.~\ref{def:input_ext}).
The system defined by $\calF\colon V\hookrightarrow\gothX(\grpG)$ is an equivariant kinematic system and every $w\in V$ can be written as
\begin{align}\label{eq:input_ext_psi}
  w = \sum_{i = 1}^{n_w} \psi_{A_i}(v_i)
\end{align}
for some $n_w \in \N$, $A_i \in \grpG$ and $v_i \in \vecV$.
\hfill$\square$
\end{theorem}

\begin{proof}
Since $\td_\star R$ is a right group action then $\td_\star R_B \td_\star R_A = \td_\star R_{AB}$ and using \eqref{eq:calF},
\begin{align*}
\td_\star R_B \calF_w & =\td_\star R_B \sum \td_\star R_{A_i} F_{v_i} \\
& = \sum \td_\star R_{A_i B} F_{v_i} \in V
\end{align*}
That is, $V$ is invariant under $\td_\star R$. Apply Lemma~\ref{lem:equivariance} to see that $\calF\colon V\hookrightarrow\gothX(\grpG)$ is equivariant. The representation \eqref{eq:input_ext_psi} is a direct consequence of \eqref{eq:calF} and \eqref{eq:Fpsi:3}.
\end{proof}

Equivariant input extensions are critical to the understanding of decompositions of general kinematic systems as they provide a structured embedding that captures the equivariant geometry of the system function even when certain velocity components are not modelled in the original measurement space $\vecV$.

The following lemma expresses equivariance in terms of an algebraic property of the lift $\Lambda$ (Def.~\ref{def:vel_lift}).

\begin{lemma}\label{lem:equiv_lift}
Let $\calF \colon V \to \gothX(\grpG)$ be an equivariant kinematic system on a Lie group $\grpG$ over a vector space $V$ with lift $\Lambda$ (Def.~\ref{def:vel_lift}).
The lift $\Lambda$ satisfies
\begin{align}\label{eq:equi_lift}
\Ad_A \Lambda(X A , \psi_A(v)) =  \Lambda(X,v)
\end{align}
for all $A, X \in \grpG$ and $v \in V$.\hfill$\square$
\end{lemma}

\begin{proof}
We have for $A, X \in \grpG$ and $v \in V$,
\begin{align}
\Ad_A & \Lambda(X A , \psi_A(v)) = \Ad_A(XA)^{-1}\calF_{\psi_A(v)}(XA) \notag \\
& = \td R_{A^{-1}}\td L_{A}\td L_{A^{-1}}\td L_{X^{-1}}\calF_{\psi_A(v)}(XA) \notag \\
& = \td L_{X^{-1}}\td R_{A^{-1}}\calF_{\psi_A(v)}(R_A(X)) \label{eq:Lambda:3} \\
& = \td L_{X^{-1}}\td R_{A^{-1}}\td R_{A}\calF_v(X) \label{eq:Lambda:4} \\
& = X^{-1}\calF_v(X) = \Lambda(X,v), \notag
\end{align}
where \eqref{eq:Lambda:3} follows because $\td L$ and $\td R$ commute and \eqref{eq:Lambda:4} follows from equivariance.
\end{proof}

\section{Invariant Systems}\label{sec:invariant_systems}

In this section we will show that while right invariant kinematic systems are also equivariant, the same does not necessarily hold for left invariant kinematic systems.\footnote{In this paper we define equivariance in terms of a right group action. If the definition is changed to a left group action, the situation with regards to left vs. right invariant systems changes accordingly.} However, we show that the equivariant input extension of a left invariant kinematic system is always left invariant.
We provide algebraic characterizations of right and left invariance in terms of conditions on the lift function $\Lambda(X,v)$, and in the equivariant case also in terms of the input action $\psi$.
As a corollary, we show that the input vector space of an invariant kinematic system can always be decomposed into subspaces that we term Type 0, Type I, or Type II corresponding to the bi-invariant, left invariant, or right invariant system components, respectively.

\begin{definition}\label{def:invariant_system} \textbf{(Invariant System)}
A kinematic system $\calF\colon V \to \gothX(\grpG)$ is
\emph{invariant} if $\image\calF\subset\gothL(\grpG)+\gothR(\grpG)$,
\emph{left invariant} if $\image\calF\subset\gothL(\grpG)$,
\emph{right invariant} if $\image\calF\subset\gothR(\grpG)$,
\emph{bi-invariant} if $\image\calF\subset\gothL(\grpG)\cap\gothR(\grpG)$, and
\emph{dual invariant} if it is invariant and
$\image\calF\cap\gothL(\grpG)\ne\{0\} \ne \image\calF\cap\gothR(\grpG)$.
\hfill$\square$\end{definition}

We start the discussion with a characterization of right invariance in terms of the lift $\Lambda$ (Def.~\ref{def:vel_lift}).

\begin{lemma}\label{lem:right_invariant}
A kinematic system $\calF\colon V \to \gothX(\grpG)$ with lift $\Lambda$ (Def.~\ref{def:vel_lift}) is right invariant if and only if $\Lambda(X,v)=\Ad_{X^{-1}}\Lambda(I,v)$ for all $X\in\grpG$ and $v\in V$.
\hfill$\square$\end{lemma}

\begin{proof}
Let the system be right invariant and let $v\in V$, then $\calF_v\in\gothR(\grpG)$ and there exists $\Lambda_v\in\gothg$ such that $\calF_v(X)=\Lambda_v X$ for all $X\in\grpG$.
Then $\Lambda(X,v)=X^{-1}\calF_v(X)=X^{-1}\Lambda_v X=\Ad_{X^{-1}}\Lambda_v$ and in particular $\Lambda(I,v)=\Lambda_v$. It follows that $\Lambda(X,v)=\Ad_{X^{-1}}\Lambda(I,v)$ for all $X\in\grpG$ and $v\in V$.

Conversely, let $\Lambda(X,v)=\Ad_{X^{-1}}\Lambda(I,v)$ for all $X\in\grpG$ and $v\in V$ then $\calF_v(X)=X\Lambda(X,v)=\td L_X \Ad_{X^{-1}}\Lambda(I,v)=\td R_X\Lambda(I,v)$ for all $X\in\grpG$ and $v\in V$ and it follows that $\calF_v\in\gothR(\grpG)$ for all $v\in V$.
\end{proof}

As a consequence of Lemma~\ref{lem:right_invariant}, we show that the right invariant kinematic systems are precisely the equivariant kinematic systems with trivial input group action.

\begin{corollary}\label{cor:right_invariant}
A kinematic system $\calF\colon V \to \gothX(\grpG)$ is right invariant if and only if it is equivariant with trivial input group action $\psi_A(v)=v$ for all $A\in\grpG$ and $v\in V$.
\hfill$\square$\end{corollary}

\begin{proof}
Let the system be right invariant and let $A\in\grpG$ and $v\in V$. By Lemma~\ref{lem:right_invariant} then
\begin{align*}
\td R_A & \calF_v(X) = \td R_A X\Lambda(X,v) = \td R_A\td L_X\Ad_{X^{-1}}\Lambda(I,v) \\
&= \td R_A\td R_X\Lambda(I,v) = \td R_{XA}\Lambda(I,v) \\
&= \td L_{XA}\td L_{(XA)^{-1}}\td R_{XA}\Lambda(I,v) \\
&= \td L_{XA}\Ad_{(XA)^{-1}}\Lambda(I,v) = XA\Lambda(XA,v) \\
&= \calF_v(R_A(X))
\end{align*}
and it follows that the system is equivariant with input group action $\psi_A(v)=v$ for all $A\in\grpG$ and $v\in V$.

Conversely, let the system be equivariant with input group action $\psi_A(v)=v$ for all $A\in\grpG$ and $v\in V$. Compute
\begin{align*}
\Lambda(X,v) &= X^{-1}\calF_v(X) \\
&= \td L_{X^{-1}}\td R_X\td R_{X^{-1}}\calF_v(R_X(I)) \\
&= \td L_{X^{-1}}\td R_X\td R_{X^{-1}}\calF_{\psi_X(v)}(R_X(I)) \\
&= \td L_{X^{-1}}\td R_X\calF_v(I) = \Ad_{X^{-1}}\Lambda(I,v)
\end{align*}
and apply Lemma~\ref{lem:right_invariant}.
\end{proof}

\begin{remark}
A simple consequence of Corollary~\ref{cor:right_invariant} is that the equivariant input extension (Def.~\ref{def:input_ext}) of a right invariant kinematic system is the system itself. That means that right invariant kinematic systems correspond precisely to the (finite dimensional) subrepresentations of the representation $\td_\star R$ of $\grpG$ on $\gothX(\grpG)$ that are entirely contained in $\gothR(\grpG)$.
\hfill$\square$\end{remark}

We now turn our attention to left invariance and give a characterization in terms of the lift $\Lambda$ (Def.~\ref{def:vel_lift}).

\begin{lemma}\label{lem:left_invariant}
A kinematic system $\calF\colon V \to \gothX(\grpG)$ with lift $\Lambda$ (Def.~\ref{def:vel_lift}) is left invariant if and only if $\Lambda(X,v)=\Lambda(I,v)$ for all $X\in\grpG$ and $v\in V$.
\hfill$\square$\end{lemma}

\begin{proof}
Let the system be left invariant and let $v\in V$, then $\calF_v\in\gothL(\grpG)$ and there exists $\Lambda_v\in\gothg$ such that $\calF_v(X)=X\Lambda_v$ for all $X\in\grpG$.
Then $\Lambda(X,v)=X^{-1}\calF_v(X)=X^{-1}X\Lambda_v=\Lambda_v$ and in particular $\Lambda(I,v)=\Lambda_v$. It follows that $\Lambda(X,v)=\Lambda(I,v)$ for all $X\in\grpG$ and $v\in V$.

Conversely, let $\Lambda(X,v)=\Lambda(I,v)$ for all $X\in\grpG$ and $v\in V$ then $\calF_v(X)=X\Lambda(X,v)=\td L_X \Lambda(I,v)$ for all $X\in\grpG$ and $v\in V$ and it follows that $\calF_v\in\gothL(\grpG)$ for all $v\in V$.
\end{proof}

Contrary to the right invariant case, left invariant systems are not always equivariant, however, we can still characterize the left invariant equivariant kinematic systems in terms of how the input group action acts at $X=I$.

\begin{corollary}\label{cor:left_invariant}
An equivariant kinematic system $\calF\colon V \to \gothX(\grpG)$ is left invariant if and only if $\Ad_{A^{-1}}\calF_v(I)=\calF_{\psi_A(v)}(I)$ for all $A\in\grpG$ and $v\in V$.
\hfill$\square$\end{corollary}

\begin{proof}
Let the equivariant kinematic system $\calF\colon V \to \gothX(\grpG)$ be left invariant and let $v\in V$. By Lemma~\ref{lem:left_invariant} then $\calF_v(X)=X\Lambda(X,v)=X\Lambda(I,v)$ for all $X\in\grpG$. Let $A\in\grpG$ then
\begin{align*}
\calF_{\psi_A(v)}(I) &= \calF_{\psi_A(v)}(R_A(A^{-1})) =\td R_A\calF_v(A^{-1}) \\
&= \td R_A A^{-1}\Lambda(I,v) = \Ad_{A^{-1}}\Lambda(I,v) \\
&= \Ad_{A^{-1}}\calF_v(I)
\end{align*}
as required.

Conversely, let $\calF\colon V \to \gothX(\grpG)$ be an equivariant kinematic system such that $\Ad_{A^{-1}}\calF_v(I)=\calF_{\psi_A(v)}(I)$ for all $A\in\grpG$ and $v\in V$.
For $X\in\grpG$ and $v\in V$ compute
\begin{align*}
\Lambda(X,v) &= X^{-1}\calF_v(X) =\td L_{X^{-1}}\td R_X\td R_{X^{-1}}\calF_v(X) \\
&= \td L_{X^{-1}}\td R_X\calF_{\psi_{X^{-1}}(v)}(R_{X^{-1}}(X)) \\
&= \Ad_{X^{-1}}\calF_{\psi_{X^{-1}}(v)}(I) = \Ad_{X^{-1}}\Ad_X\calF_v(I) \\
&= \calF_v(I) = \Lambda(I,v)
\end{align*}
and apply Lemma~\ref{lem:left_invariant}.
\end{proof}

While a left invariant kinematic system need not be equivariant, its equivariant input extension (Def.~\ref{def:input_ext}) is also left invariant as the following lemma shows.

\begin{lemma}\label{lem:left_invariant_input_ext}
Let $F\colon \vecV \to \gothX(\grpG)$ be a kinematic system on a Lie group $\grpG$ over a vector space $\vecV$. If the system is left invariant then its equivariant input extension (Def.~\ref{def:input_ext}) is left invariant.
\hfill$\square$\end{lemma}

\begin{proof}
Let the kinematic system $F\colon \vecV \to \gothX(\grpG)$ be left invariant and let $\calF\colon V \to \gothX(\grpG)$ be its equivariant input extension. Let $A,X\in\grpG$ and $v\in V$ then $\psi_A(v)\in V$ and there exist $A_i\in\grpG$, $v_i\in\vecV$ and $\Lambda_{v_i}\in\gothg$ such that
\begin{align*}
\calF_{\psi_A(v)}(X) & = \sum \td_\star R_{A_i} F_{v_i}(X) \\
& = \sum \td R_{A_i} F_{v_i}(XA_i^{-1}) \\
& = \sum \td R_{A_i} XA_i^{-1}\Lambda_{v_i}.
\end{align*}
Define $\Gamma_i:=A_i^{-1}\Lambda_{v_i}A_i\in\gothg$ then
\begin{align*}
\calF_{\psi_A(v)}(X) = \sum X\Gamma_i.
\end{align*}
Compute
\begin{align*}
\Ad_{A^{-1}} \calF_v(I) &= \td L_{A^{-1}} \td R_A \calF_v(I) \\
& = \td L_{A^{-1}} \calF_{\psi_A(v)}(R_A(I)) \\
& = \td L_{A^{-1}} \sum A\Gamma_i = \sum \Gamma_i = \calF_{\psi_A(v)}(I),
\end{align*}
and it follows from Corollary~\ref{cor:left_invariant} that the extended system $\calF\colon V \to \gothX(\grpG)$ is left invariant.
\end{proof}

\begin{remark}
The above analysis shows that left invariant equivariant kinematic systems correspond precisely to the (finite dimensional) subrepresentations of the representation $\td_\star R$ of $\grpG$ on $\gothX(\grpG)$ that are entirely contained in $\gothL(\grpG)$.
\hfill$\square$\end{remark}

Putting these results together, we arrive at the following characterization of bi-invariant kinematic systems.

\begin{corollary}\label{cor:bi-invariant}
A kinematic system $\calF\colon V \to \gothX(\grpG)$ with lift $\Lambda$ (Def.~\ref{def:vel_lift}) is bi-invariant if and only if $\Lambda(X,v)=\Ad_{X^{-1}}\Lambda(I,v)=\Lambda(I,v)$ for all $X\in\grpG$ and $v\in V$ or, equivalently, if it is equivariant with trivial input group action $\psi_A(v)=v$ and $\Ad_{A^{-1}}\calF_v(I)=\calF_v(I)$ for all $A\in\grpG$ and $v\in V$.
\hfill$\square$\end{corollary}

Using the above results we can now derive a structure theorem for invariant kinematic systems in terms of a direct sum decomposition of the input vector space.

\begin{definition}\label{def:type} \textbf{(Type 0, Type I and Type II)}
Let $\calF\colon V \to \gothX(\grpG)$ be a kinematic system with lift $\Lambda$ (Def.~\ref{def:vel_lift}).
\begin{List}
\item[Type I:]  A \emph{Type I invariant velocity subspace} is a subspace $V^I \subset V$ such that for all $X\in\grpG$ and $v\in V^I$,
    \begin{align}
      \Lambda(X,v) = \Lambda(\Id,v).
      \label{eq:Type_I}
    \end{align}

\item[Type II:]  A \emph{Type II invariant velocity subspace} is a subspace $V^{II} \subset V$ such that for all $X\in\grpG$ and $v\in V^{II}$,
    \begin{align}
      \Lambda(X,v) = \Ad_{X^{-1}}\Lambda(\Id,v).
      \label{eq:Type_II}
    \end{align}

\item[Type 0:]  A \emph{Type 0 invariant velocity subspace} is a subspace $V^{0}\subset V$ such that for all $X\in\grpG$ and $v\in V^{0}$ both \eqref{eq:Type_I} and \eqref{eq:Type_II} hold.\hfill$\square$
\end{List}
\end{definition}

The following is then a simple corollary of Lemma~\ref{lem:right_invariant}, Lemma~\ref{lem:left_invariant} and Corollary~\ref{cor:bi-invariant}.

\begin{corollary}\label{cor:inv_sys}
Let $\calF \colon V \to \gothX(\grpG)$ be a kinematic system with lift $\Lambda$  (Def.~\ref{def:vel_lift}). Let $V^{0}$, $V^I$, $V^{II} \subset V$ denote Type 0, Type I and Type II invariant velocity subspaces, respectively.
\begin{enumerate}
\item  The kinematic system $\calF^0 \colon V^0 \to \gothB(\grpG) \subset \gothX(\grpG)$ obtained by restricting $\calF$ to $V^0$ is bi-invariant. The velocity space $V^0$ is finite dimensional with $\dim V^0\leq\dim\grpG$.

\item The kinematic system $\calF^I \colon V^I \to \gothL(\grpG) \subset \gothX(\grpG)$ obtained by restricting $\calF$ to $V^I$ is left invariant.  The velocity space $V^I$ is finite dimensional with $\dim V^I\leq\dim\grpG$.

\item The kinematic system $\calF^{II} \colon V^{II} \to \gothR(\grpG) \subset \gothX(\grpG)$ obtained by restricting $\calF$ to $V^{II}$ is right invariant. The velocity space $V^{II}$ is finite dimensional with $\dim V^{II}\leq\dim\grpG$.
\end{enumerate}
Conversely, if $\vecV\subset V$ is a vector subspace and the kinematic system $F\colon\vecV\to\gothX(\grpG)$ obtained by restricting $\calF$ to $\vecV$ is bi-invariant (resp. left invariant resp. right invariant) then $\vecV$ is a Type 0 (resp. Type I resp. Type II) invariant velocity subspace.
\hfill$\square$\end{corollary}

Using Corollary~\ref{cor:right_invariant} and Corollary~\ref{cor:inv_sys}, we immediately obtain an alternative characterization of Type II invariant velocity subspaces for an equivariant kinematic system in terms of the input action $\psi$.

\begin{corollary}\label{cor:type_psi}
For an equivariant kinematic system, \eqref{eq:Type_II} is equivalent to
\begin{align}\label{eq:Type_II_psi}
  \psi_A(v)=v
\end{align}
for all $A\in\grpG$ and $v\in V^{II}$.
\hfill$\square$\end{corollary}

As a consequence of Corollary~\ref{cor:inv_sys}, we obtain the following structure theorem for invariant kinematic systems.

\begin{theorem}\label{th:inv_sys_decomp}
Let $\calF \colon V \to \gothX(\grpG)$ be an invariant kinematic system on a Lie group $\grpG$ over a vector space $V$.
Then $V$ is finite dimensional and there is a direct sum decomposition of $V$ into
\[
V = V^0 \oplus V^I_\perp \oplus V^{II}_\perp,
\]
where $V^0$ is a Type 0, $V^I_\perp$ a Type I and $V^{II}_\perp$ a Type II invariant velocity subspace.
\end{theorem}

\begin{proof}
Since the system is invariant, we have $\image\calF\subset\gothL(\grpG)+\gothR(\grpG)$.
Let $V^I\subset V$ (resp. $V^{II}\subset V$) be the preimage of $\gothL(\grpG)$ (resp. $\gothR(\grpG)$) under $\calF$ then $V=V^I+V^{II}$ by linearity of $\calF$. By Corollary~\ref{cor:inv_sys}, $V^I$ is a Type I invariant velocity subspace and $V^{II}$ is a Type II invariant velocity subspace. Define $V^0 := V^I \cap V^{II}$ then $V^0$ is a Type 0 invariant velocity subspace.

By Corollary~\ref{cor:inv_sys}, $V^I$ and $V^{II}$ (and hence $V$) are finite dimensional and one can find direct sum decompositions $V^I = V^I_\perp \oplus V^0$ of $V^I$ and $V^{II} = V^{II}_\perp \oplus V^0$ of $V^{II}$. The Type I resp. Type II property of $V^I_\perp$ resp. $V^{II}_\perp$ follows from Corollary~\ref{cor:inv_sys}. We now have $V = V^0 \oplus V^I_\perp \oplus V^{II}_\perp$.
\end{proof}

An interesting consequence of Theorem~\ref{th:inv_sys_decomp} is that if a kinematic system is invariant then $\dim V \leq 2 \dim \grpG$ dimensions.
Interestingly, the kinematics on $\grpG$ have at most $\dim \grpG$ instantaneously independent degrees of freedom, so it may be the case that the velocities $v \in V$ are instantaneously dependent.
For example, a measurement of airspeed derived from a pitot tube measurement device on an aerial robot is a body-fixed frame measurement of linear velocity and is a Type I invariant velocity in the above language [\cite{RM_2013_Mahony_nolcos}].
A measurement of GPS velocity is a reference-fixed measurement of linear velocity and is a Type II invariant velocity in the above language [\cite{RM_2013_Mahony_nolcos}].
Both measurements concern the same physical velocity and are instantaneously dependent, however, clearly the measurements themselves are not identical.
This example emphasises the importance of the structural analysis of the input space, providing a mechanism to understand and separate the velocities associated with left-invariant (or body-fixed frame) velocity measurements from right-invariant (or reference fixed frame) velocity measurements.
Importantly, the structure of $V$ reflects the nature of the measurement processes and is not just a reparameterization of the tangent space of $\grpG$.

\section{Group Affine Systems}\label{sec:GA_systems}

The class of group affine systems was introduced recently to provide a condition for suitability of a system for observer design using the Invariant Extended Kalman Filter (IEKF) framework [\cite{2016_Barrau_arxive,2017_Barrau_tac}].

\begin{definition}\label{def:GA_sys}\textbf{(Group Affine System)} [\cite{2017_Barrau_tac}]
A kinematic system $\calF\colon V \rightarrow \gothX(\grpG)$ is \emph{group affine} if it satisfies
\begin{align}
\calF_v(A B) = \calF_v(A) B + A \calF_v(B) - A \calF_v(\Id) B
\label{eq:GA_def}
\end{align}
for all $A,B\in\grpG$ and $v\in V$.
\hfill$\square$
\end{definition}

We start by showing that invariant kinematic systems are group affine.

\begin{lemma}\label{lem:inv_implies_GA}
Let $\calF \colon V \to \gothX(\grpG)$ be an invariant kinematic system on a Lie group $\grpG$ over a vector space $V$. Then the system is group affine.
\hfill$\square$
\end{lemma}

\begin{proof}
Since the system is invariant, all velocities are Type I or Type II invariant. Consider the lift $\Lambda$ (Def.~\ref{def:vel_lift}) and consider the case that $v\in V$ is Type I invariant.
Substituting $\calF_v(X) = X\Lambda(X,v)=X\Lambda(\Id,v)$, cf. \eqref{eq:Type_I}, into the right hand side of \eqref{eq:GA_def} yields
\begin{align*}
\calF_v(A) B +&  A \calF_v(B) - A \calF_v(\Id) B   \\
& = A\Lambda(\Id,v) B + AB\Lambda(\Id,v) - A \Lambda(\Id,v) B  \\
& = AB\Lambda(\Id,v) = F_v(AB)
\end{align*}
which proves \eqref{eq:GA_def} in this case.
For $v\in V$ Type II invariant then substituting $F_v(X) = X\Lambda(X,v)=\Lambda(\Id,v)X$, cf. \eqref{eq:Type_II}, into the right hand side of \eqref{eq:GA_def} yields
\begin{align*}
\calF_v(A) B +&  A \calF_v(B) - A \calF_v(\Id) B  \\
& = \Lambda(\Id,v) AB + A\Lambda(\Id,v)B - A \Lambda(\Id,v) B  \\
& =\Lambda(\Id,v) AB = F_v(AB)
\end{align*}
which proves \eqref{eq:GA_def} also in this case.
\end{proof}

It follows from Lemma~\ref{lem:inv_implies_GA} and the discussion in Section~\ref{sec:invariant_systems} that a group affine kinematic system is not necessarily equivariant.
We now prove that the equivariant input extension of a group affine kinematic system is again group affine and differs from the original system only by a (finite dimensional) Type I invariant velocity subspace.

\begin{theorem}\label{th:GA_ext}
Let $F \colon \vecV \to \gothX(\grpG)$ be a group affine kinematic system on a Lie group $\grpG$ over a vector space $\vecV$ and let $\calF \colon V \to \gothX(\grpG)$ be its equivariant input extension.
Then there exists a Type I invariant velocity subspace $V_\perp^{I}\subset V$ such that $V=\vecV \oplus V_\perp^{I}$ and $\calF \colon V \to \gothX(\grpG)$ is group affine.
\hfill$\square$\end{theorem}

\begin{proof}
Consider the lift $\Lambda$ (Def.~\ref{def:vel_lift}) of the equivariant input extension
$\calF \colon V \to \gothX(\grpG)$ and recall that $\calF_v=F_v$ for $v\in\vecV\equiv\image F\subset V$. Hence the lift of $F$ is given by the restriction of $\Lambda$ to $\grpG\times\vecV$. Substituting into \eqref{eq:GA_def} one obtains
\begin{gather*}
AB \Lambda (AB, v) = A \Lambda(A,v) B + AB \Lambda(B,v) - A \Lambda(\Id,v) B
\end{gather*}
for all $A,B\in\grpG$ and $v\in\vecV$.
Pre-multiplying by $A^{-1}$ and post-multiplying by $B^{-1}$ yields
\begin{gather*}
B \Lambda (AB, v)B^{-1} -\Lambda(A,v) =  B \Lambda(B,v)B^{-1} - \Lambda(\Id,v),
\end{gather*}
using equivariance of the lift function \eqref{eq:equi_lift} yields
\begin{gather*}
\Lambda (A, \psi_{B^{-1}}( v)) -\Lambda(A,v) =   \Lambda(\Id,\psi_{B^{-1}}( v)) - \Lambda(\Id,v),
\end{gather*}
and finally exploiting linearity of the lift with respect to the velocity yields
\begin{gather}
\Lambda (A, \psi_{B^{-1}}(v) - v) =   \Lambda(\Id, \psi_{B^{-1}}( v)-v)
\label{eq:GA_lift}
\end{gather}
for all $A,B\in\grpG$ and $v\in\vecV$.

Define a family of linear operators
\[
D^\psi_\grpG = \{ D^\psi_B \colon V \to V,\  D^\psi_B(v) := \psi_{B^{-1}}(v) - v \;|\;  B \in \grpG \}
\]
where $D$ stands for the difference of the linear operator $\psi_{B^{-1}}$ to the identity operator $\Id(v) = v$ acting on $V$.
Consider the image of $\vecV\subset V$ under the family $D^\psi_\grpG$, that is
\[
D^\psi_\grpG(\vecV) = \{ w \in V \;|\; w = D^\psi_B(v), B \in \grpG, v \in \vecV \} .
\]
For any $w=D^\psi_B(v)\in D^\psi_\grpG(\vecV)$ (noting that $v \in \vecV$) then \eqref{eq:GA_lift} shows that
\begin{align}\label{eq:image_LG}
\Lambda (A, w ) &=   \Lambda(\Id, w), \text{ for all } A \in \grpG.
\end{align}
That is, by definition and using that \eqref{eq:image_LG} is linear in $w$,
\[
V^I := \text{span} D^\psi_\grpG(\vecV)\subset V
\]
is a Type I invariant velocity subspace and hence finite dimensional with $\dim V^{I}\leq\dim\grpG$ by Corollary~\ref{cor:inv_sys}.

To show that $V = \vecV + V^I$ we only need to show that $V \subset \vecV + V^I$. Let $w\in V$. By \eqref{eq:input_ext_psi}, there exist $B_i \in \grpG$ and $v_i \in \vecV$ such that $w = \sum \psi_{B_i}(v_i)$. Then
\begin{align*}
w - \sum v_i =  \sum \psi_{B_i}(v_i) - \sum v_i = \sum D^\psi_{B_i}(v_i),
\end{align*}
that is $w = \sum v_i + \sum D^\psi_{B_i}(v_i) \in \vecV + V^I$.
It follows that $V = \vecV + V^I$.

Since $V^{I}$ is finite dimensional one can find a direct sum decomposition $V^{I}=(\vecV \cap V^{I})\oplus V_\perp^{I}$. The Type I property of $V_\perp^{I}$ follows from Corollary~\ref{cor:inv_sys}. We now have $V=\vecV \oplus V_\perp^{I}$.

The extended system $\calF \colon V \to \gothX(\grpG)$ is group affine because for every $w\in V$ there exist $v\in\vecV$ and $v^{I}\in V_\perp^{I}$ such that $w=v+v^{I}$, and as \eqref{eq:GA_def} is linear in $v$ and holds for $v\in\vecV$, we only need to check that it also holds for $v^{I}\in V_\perp^{I}$. The result then follows from Corollary~\ref{cor:inv_sys} and Lemma~\ref{lem:inv_implies_GA}.
\end{proof}

The study of the fine structure of group affine systems (in the spirit of Theorem~\ref{th:inv_sys_decomp}) requires advanced Lie theoretic tools and is beyond the scope of this paper.

\section{Equivariant Filter (\texorpdfstring{E\MakeLowercase{q}F}{EqF})}
\label{sec:EqF}

In this section we present a new filter design that we term the \emph{Equivariant Filter} or EqF.
This filter can be applied to any equivariant kinematic system on a Lie group, and in general, to any system that can be embedded into such an equivariant system. By Theorem~\ref{th:equiv_ext}, this includes any kinematic system (equivariant or not) on a Lie group.

In order to introduce the filter design then it is necessary to define state measurements.
Assume that there is a measurement process $h\colon \grpG \to \R^m$ where for simplicity in the present paper we assume that the output is the real vector space of dimension $m$.
That is, we model a nonlinear measurement process
\[
y = h(X) \in \R^m.
\]

Define \emph{wedge} and \emph{vee} linear operators $\wedge\colon \R^n \to \gothg$ and $\vee\colon \gothg \to \R^n$ that map the Lie algebra $\gothg$ back and forth into a real vector space $\R^n$, where $n = \dim\gothg = \dim \grpG$.
Thus, for $\epsilon \in \gothg$, $\epsilon^\vee \in \R^n$ while $\epsilon = (\epsilon^\vee)^\wedge \in \gothg$.
For a linear operator $A_t\colon \gothg \to \gothg$ then $A_t^\vee \in \R^{n \times n}$ is the matrix operator on $\R^n$ such that
\[
A_t(\epsilon) = \left( A_t^\vee \epsilon^\vee \right)^\wedge
\]
for all $\epsilon\in\gothg$.
Similarly, for a linear operator $C_t\colon \gothg \to \R^m$,
\[
C_t (\epsilon) = C_t^\vee \epsilon^\vee
\]
for all $\epsilon\in\gothg$, where the fact that we map into a vector output space obviates the need for  wedge and vee operators on the output space.

The observer state $\hat{X} \in \grpG$ is taken as an element of the Lie group $\grpG$.
The canonical state error is
\[
E =  X \hat{X}^{-1}.
\]
We will consider an observer of the form
\begin{align}
\ddt{\hat{X}} = \calF_{v(t)} (\hat{X}) + \Delta_t \hat{X},
\label{eq:observer}
\end{align}
where $\Delta_t \in \gothg$ is a correction term that will be derived from the output measurements $y$ and the observer state $\hat{X}$ and will depend on time varying gains $\Sigma$ and the linearisation discussed below.
The error kinematics are
\begin{align*}
\dot{E} &= \dot{X}\hat{X}^{-1}-X\hat{X}^{-1}(\ddt{\hat{X}})\hat{X}^{-1} \\
&= X\Lambda(X,v)\hat{X}^{-1}-X\Lambda(\hat{X},v)\hat{X}^{-1}-X\hat{X}^{-1}\Delta_t \\
&= E \Ad_{\hat{X}} \left( \Lambda(X, v) - \Lambda(\hat{X}, v) \right)
- E \Delta_t,
\end{align*}
expressed in terms of the lift $\Lambda$ (Def.~\ref{def:vel_lift}).
Since we assume equivariance, then Lemma~\ref{lem:equiv_lift} allows the simplification
\begin{align}
\dot{E} = E  \left( \Lambda(E, \psi_{\hat{X}^{-1}}(v)) - \Lambda(\Id, \psi_{\hat{X}^{-1}}(v)) \right)
- E \Delta_t.
\label{eq:equi_error_dyn}
\end{align}
The key observation for these error kinematics is that apart from an observer state dependent transformation of the input signal, and indeed the time dependence of the input signal itself, the dynamics are autonomous in the error $E$.
That is, setting $w(t) = \psi_{\hat{X}^{-1}}(v(t)) \in V$ to be a known exogenous input one has
\begin{align}
\dot{E} = E  \left( \Lambda(E, w(t)) - \Lambda(\Id, w(t)) \right)
- E \Delta_t.\notag
\end{align}

\begin{remark}
If the system is equivariant and group affine then \eqref{eq:GA_lift} can be rewritten as
\[
\Lambda (A, \psi_{B^{-1}}(v)) - \Lambda(\Id, \psi_{B^{-1}}( v))
= \Lambda (A, v) - \Lambda(\Id,v)
\]
for all $A,B\in\grpG$ and $v\in V$.
From this it is clear that the error kinematics \eqref{eq:equi_error_dyn} can be written directly as
\[
\dot{E} = E  \left( \Lambda(E, v(t)) - \Lambda(\Id, v(t)) \right)
- E \Delta_t
\]
in the group affine case, and they are explicitly independent of the observer state.
This formula does not depend on the equivariant extension or the definition of $\psi$, meaning that it also holds for a group affine system that is not equivariant.
It is this structure of the error kinematics that is exploited by the Invariant Extended Kalman Filter, cf. \cite{2017_Barrau_tac}.
\hfill$\square$
\end{remark}

The proposed filter design is based on the linearisation of \eqref{eq:equi_error_dyn} at the identity $\Id\in\grpG$, taking the input $w(t) = \psi_{\hat{X}^{-1}}(v(t))$ as an exogenous input.
We denote the linearisation of the state error $E$ at $\Id \in \grpG$ by $\epsilon$.
That is, $E \approx \Id + \epsilon$ at least to first order.
It is straightforward to see that $\epsilon \in \gothg$ is an element of the identity tangent space of $\grpG$.
The linearisation of the state error kinematics is given by the differential of $\Lambda(E,w(t))$ with respect to $E$ at the origin
\[
A_t := \left(  \left. \tD_E \right|_{\Id} \Lambda(E,w(t)) \right)^\vee
\]
and is time varying depending on the signal $w(t)=\psi_{\hat{X}^{-1}(t)} (v(t))$.
Note that in general there is an algebraic formula available for differentiation of $\Lambda(E,v)$ but no direct algebraic form for $\Lambda(E, \psi_{\hat{X}^{-1}} (v))$, because the latter is constructed through an equivariant input extension.
However, exploiting equivariance and using Lemma~\ref{lem:equiv_lift} one has
\begin{align*}
\left. \tD_E \right|_{\Id} \Lambda(E,\psi_{\hat{X}^{-1}} (v))
& = \left. \tD_E \right|_{\Id} \Ad_{\hat{X}} \Lambda(R_{\hat{X}}E,v) \\
& = \Ad_{\hat{X}} \left. \tD_E \right|_{\hat{X}} \Lambda(E,v)
\left.  \tD \right|_{\Id} R_{\hat{X}},
\end{align*}
where we use linearity of $\Ad_{\hat{X}}$ and the chain rule to obtain an algebraic formula for the linearisation in an explicit form.

The linearisation of the output error equation is given by
\[
C_t := \left( \left. \tD_E \right|_{\hat{X}} h \circ \td R_{\hat{X}} \right)^\vee,
\]
and the linearised error dynamics are
\begin{align}
\dot{\epsilon} & = A_t \epsilon - \Delta_t,  \notag \\
y & = C_t \epsilon. \notag
\end{align}

Let $P \in \PD(n)$ and $Q \in \PD(m)$ be positive definite matrices then the correction term $\Delta_t$ is designed based on a standard Kalman-Bucy filter design for the linearisation with state noise process $N(0,P)$ and output noise process $N(0,Q)$.
Substituting the resulting correction term into the observer equation \eqref{eq:observer} yields
\begin{align}
\ddt \hat{X} = \calF_{v(t)} (\hat{X}) + \td R_{\hat{X}} \left(\Sigma C_t^\top Q^{-1} (y - h(\hat{X})) \right)^\wedge,
\label{eq:EqF}
\end{align}
where $\Sigma \in \R^{n \times n}$ is the time varying solution of the Riccati equation
\begin{align}
\ddt \Sigma = \Sigma A_t + A_t^\top \Sigma + P - \Sigma C_t^\top Q C_t \Sigma.
\label{eq:EqF_Sigma}
\end{align}

The EqF filter \eqref{eq:EqF} and \eqref{eq:EqF_Sigma} specialises directly to the highly successful IEKF in the case where the system is group affine or invariant.
The proposed formulation now brings this to a broader class of systems, including any kinematic system on a Lie group.

\section{Conclusion}\label{sec:conclusion}

We have shown that any kinematic system on a Lie group can be embedded into an equivariant kinematic system through an extension of the input space. Restricting the new input to the original input space yields the original system behaviour. This structured embedding allows the construction of a general Equivariant Filter (EqF) for kinematic systems on Lie groups that specialises to the highly successful Invariant Extended Kalman Filter (IEKF) in the case of group affine or invariant kinematic systems.

\section*{Acknowledgments}
This research was supported by the Australian Research Council
through the ``Australian Centre of Excellence for Robotic Vision'' CE140100016
and through the Discovery Project DP160100783 ``Sensing a complex world: Infinite dimensional observer theory for robots''.



\end{document}